\documentclass[11pt,a4paper]{article}

\usepackage[margin=1.05in]{geometry}
\usepackage{amsmath,amssymb,amsthm}
\usepackage{mathrsfs}
\usepackage{xcolor}

\usepackage[colorlinks=true,linkcolor=blue,citecolor=blue,urlcolor=blue]{hyperref}

\newcommand{\vN}{von~Neumann}

\newcommand{\alg}[1]{\mathcal{#1}}
\newcommand{\Hilb}{\mathcal{H}}

\newcommand{\artanh}{\operatorname{artanh}}

\theoremstyle{plain}
\newtheorem{theorem}{Theorem}
\newtheorem{proposition}[theorem]{Proposition}
\newtheorem{lemma}[theorem]{Lemma}
\newtheorem{corollary}[theorem]{Corollary}
\theoremstyle{remark}
\newtheorem{remark}{Remark}

\title{The modular energy range of an interval\\
in chiral conformal field theory}

\author{
Samuel L.~Braunstein$^{1}$ \and Zhi-Wei Wang$^{2,1}$
}

\date{%
\small
$^{1}$Department of Computer Science, University of York, York YO10 5GH, United Kingdom\\
$^{2}$College of Physics, Jilin University, Changchun 130012, China\\[1ex]
Corresponding author: S.~L.~Braunstein,
\texttt{sam.braunstein@york.ac.uk}\\
{ORCID: S.~L.~Braunstein 0000-0003-4790-136X;
Z.-W.~Wang 0000-0003-0847-1029}\\[2ex]
\today}

\begin{document}

\maketitle

\begin{abstract}
We determine exactly the range of the modular Hamiltonian of an interval
in a chiral conformal field theory over states of bounded energy.  For an
interval of length $R$ in a theory of central charge $c$, and states of
mean energy at most $E$, the supremum of the modular energy is
\[
  M_{\max}(E,R)
  = \min_{\beta>0}\Big\{\! (1+\beta)\tfrac{\pi R}{2}E
   + \tfrac{c}{6}\big[1-\sqrt\beta\arctan\tfrac1{\sqrt\beta}\big]\!\Big\}
  = \frac{\pi R}{2}E + \frac{c}{6}
   - \frac{\pi c^{2}}{288\,RE} + O\big((RE)^{-2}\big) ,
\]
and neither the coefficient $\pi R/2$ nor the additive constant $c/6$ can
be improved.  Both are exact rather than merely bounds: the
one-parameter family of weights used to prove them is the Legendre
structure of the answer, and the energy of the optimizing state is
identically minus the derivative of the bound with respect to the
additive shift.  The infimum is logarithmic,
$M_{\min} \simeq -\tfrac{c}{12}\big[\ln(24\pi RE/c)-1\big]$, so that the
range is linear above and logarithmic below, the two meeting in the limit
of small $RE$.  It is strictly wider above than below, by more than
$\tfrac{2}{5}\pi RE$ at every energy, a statement that reduces after a
Stieltjes substitution to the positivity of a variance.

The two ingredients of the upper bound separate cleanly.
The coefficient $\pi R/2$ is a consequence of M\"obius representation
theory alone: the weight complementary to the modular one is a translate
of the special conformal generator and therefore positive, so for the
full modular Hamiltonian the bound holds with no additive constant.  The
constant $c/6$ is precisely the price of truncating that weight at the
endpoints of the interval, and it is supplied by the quantum energy
inequality of Fewster and Hollands.  The state that saturates the bound
places its negative energy exactly at the two endpoints, which is the
same statement seen from the other side.  Accordingly the constant
is additive over entangling points: for $n$ disjoint intervals of equal
length it is $nc/6$, or $c/12$ per point.

Composing the bound with positivity of relative entropy gives a
sharp Bekenstein bound, $\Delta S_I \leq \frac{\pi R}{2}E + \frac{c}{6}$,
which recovers the original form of that inequality from the modular form
that by itself imposes no constraint.
\end{abstract}



\section{Introduction}
\label{sec:intro}

The modular Hamiltonian of a spacetime region has become a central object
in quantum field theory.  For a wedge it generates the boost, by the
theorem of Bisognano and Wichmann \cite{BisognanoWichmann1975}; for a
conformal theory it is known explicitly for balls and intervals
\cite{HislopLongo1982,CHM2011}; and it controls relative entropy, the
Bekenstein bound \cite{Casini2008,BlancoCasini2013}, and a growing body of
work on modular flow and its geometry.  Its relation to the ordinary
energy is therefore a natural question, and the algebraic side of the
subject answers a version of it exactly: by Borchers' theorem
\cite{Borchers1992} the modular group of a wedge acts on the translations
as a dilation,
$\mathrm{Ad}\,\Delta^{it}\,U(a) = U(e^{-2\pi t}a)$.

That is an identity between group actions.  Here we ask a quantitative
question instead.  \emph{How large can the modular energy of a region be,
on states whose ordinary energy is bounded?}  For a chiral conformal
field theory and an interval we answer it exactly, and the answer has a
structure worth stating before the details.

Write $I = (a,b)$ for an interval of length $R = b-a$ on the light ray,
$T$ for the chiral stress tensor, $P = \int T$ for the translation
generator, and
\begin{equation}
  K_I \;=\; 2\pi\int_a^b \frac{(x-a)(b-x)}{R}\,T(x)\,dx
  \label{eq:KIintro}
\end{equation}
for the modular weight of $I$.  The weight is bounded by $\pi R/2$, and
one is tempted to conclude that $K_I \leq \frac{\pi R}{2}P$ as an
operator inequality.  The inference is not available.  The difference is
a smearing of $T$ against a non-negative function, and such smearings are
not positive operators: local energy densities take negative expectation
values, by the theorem of Epstein, Glaser and Jaffe \cite{EGJ1965}.  Only
the total energy is protected by the spectrum condition.

Nevertheless the conclusion is correct up to an additive constant, and
the two ingredients enter separately.  For the \emph{full} modular
Hamiltonian, whose weight extends over the whole line, the bound holds
exactly and with no constant, because the complementary weight is a
perfect square and hence lies in the M\"obius positive cone of Blanco and
Casini \cite{BlancoCasini2013}: this is Proposition~\ref{prop:mobius}
below, and it accounts for the coefficient $\pi R/2$ without any
appeal to energy inequalities.  For the one-sided $K_I$ of
\eqref{eq:KIintro} the weight must be truncated at the endpoints of $I$,
and the price of that truncation is exactly $c/6$, supplied by the
quantum energy inequality of Fewster and Hollands
\cite{FewsterHollands2005}.  That the saturating state puts its negative
energy precisely at the two endpoints, where the modular weight vanishes,
is the same fact seen from the other side.

The results are Theorem~\ref{thm:main} for the supremum and
Theorem~\ref{thm:bottom} for the infimum.  Both are exact rather than
bounds: the one-parameter family of weights used to prove them is not a
technical device but the Legendre structure of the answer, as
Remark~\ref{rem:legendre} explains.  To our knowledge the operator
inequality \eqref{eq:endpoint} is new; Section~\ref{sec:disc} discusses
the nearest results.

Section~\ref{sec:setting} fixes the setting.
Section~\ref{sec:mobius} proves the M\"obius bound for the full modular
Hamiltonian, and Section~\ref{sec:qei} assembles the quantum energy
inequality in the form needed, including a lemma on admissible weights.
Section~\ref{sec:optimizer} isolates an identity for the energy of the
optimizing state which makes both ranges exact rather than bounded.
Sections \ref{sec:top} and \ref{sec:bottom} compute the top and the
bottom, Section~\ref{sec:multi} treats several intervals, and
Section~\ref{sec:disc} discusses the relation to existing bounds.

\section{Chiral conformal nets}
\label{sec:setting}

We work with a M\"obius covariant local net on the light ray, in the
standard sense \cite{GabbianiFrohlich1993,KawahigashiLongo2004}.  Thus
$I \mapsto \alg{A}(I)$ assigns a \vN\ algebra on a Hilbert space $\Hilb$
to each bounded open interval, isotonically, with $\alg{A}(I_1)$ and
$\alg{A}(I_2)$ commuting for disjoint intervals; there is a unitary
positive-energy representation of the M\"obius group under which the net
is covariant; and the vacuum $\Omega$ is the unique invariant vector.  We
assume the theory has a stress tensor $T$ of central charge $c$, so that
the translation generator is
\begin{equation}
  P \;=\; \int_{\mathbb{R}} T(x)\,dx \;\geq\; 0 ,
  \qquad P\,\Omega = 0 ,
  \label{eq:P}
\end{equation}
positivity being the spectrum condition.  We write $\alg{D}$ for the
dense domain of vectors on which the smeared stress tensor is defined and
for which $\langle T(v)\rangle_\psi$ is smooth with
$\langle T(v)\rangle_\psi = O(v^{-4})$ as $|v| \to \infty$; this is the
domain of \cite{FewsterHollands2005}, the decay being their consequence
of axiom~B.3, the smoothness of $\langle\Theta(z)\rangle_\psi$ on the
circle, together with their Eq.~(3.38).  In particular
$\langle T(\cdot)\rangle_\psi \in L^1(\mathbb{R})$ for $\psi \in \alg{D}$,
and we write $\Lambda_\psi$ for its $L^1$ norm.

For an interval $I = (a,b)$ of length $R$ put
\begin{equation}
  w(x) \;=\; \frac{(x-a)(b-x)}{R} ,
  \qquad
  K_I \;=\; 2\pi \int_a^b w(x)\,T(x)\,dx ,
  \label{eq:KI}
\end{equation}
the modular weight obtained from Bisognano-Wichmann by conformal
transport \cite{BisognanoWichmann1975,HislopLongo1982}.  The weight
vanishes at the endpoints and attains $R/4$ at the midpoint, so
\begin{equation}
  0 \;\leq\; 2\pi w(x) \;\leq\; \frac{\pi R}{2} .
  \label{eq:wmax}
\end{equation}

\begin{remark}[Normalization]
\label{rem:normalization}
Modular Hamiltonians are fixed by the flow they generate and hence only
up to an additive constant, which in field theory is cutoff dependent.
Equation~\eqref{eq:KI} removes that freedom: $K_I$ is defined as an
explicit integral of the stress tensor, and since
$\langle T\rangle_\Omega = 0$ it satisfies
$\langle K_I\rangle_\Omega = 0$.  The constant $c/6$ below is therefore a
feature of the inequality and not a residue of a convention, and the two
additive constants should not be conflated.
\end{remark}

\begin{remark}[One-sided and full]
\label{rem:naming}
Equation \eqref{eq:KI} integrates over $I$ only, so $K_I$ is the
one-sided modular weight.  The Tomita-Takesaki generator of
$(\alg{A}(I),\Omega)$ is $-\log\Delta_I = \hat K_I = K_I - K_{I'}$, the
same weight extended over the whole line.  We treat both, and they behave
quite differently: Section~\ref{sec:mobius} bounds $\hat K_I$ with no
additive constant, and Section~\ref{sec:top} shows that $K_I$ costs
$c/6$.
\end{remark}

\section{The full modular Hamiltonian: a M\"obius bound}
\label{sec:mobius}

\begin{proposition}
\label{prop:mobius}
With $x_c = (a+b)/2$ and
$G_{x_c} = \int_{\mathbb{R}}(x-x_c)^2\,T(x)\,dx$,
\begin{equation}
  \hat K_I
  \;=\; \frac{2\pi}{R}\int_{\mathbb{R}}(x-a)(b-x)\,T(x)\,dx
  \;=\; \frac{\pi R}{2}\,P \;-\; \frac{2\pi}{R}\,G_{x_c} ,
  \label{eq:mobius}
\end{equation}
and since $G_{x_c} \geq 0$ we have $\hat K_I \leq \frac{\pi R}{2}P$, with
no additive constant and no appeal to any energy inequality.
\end{proposition}

\begin{proof}
The weights satisfy the identity
$(x-a)(b-x) = R^2/4 - (x-x_c)^2$.  The operator $G_{x_c}$ is a translate
of the generator of special conformal transformations, unitarily
equivalent to $P$ and therefore positive.  More generally, for
$\delta,\gamma \geq 0$ the quadratic $h(x) = \delta + \gamma(x-x_0)^2$ is
non-negative and $\int h\,T = \delta P + \gamma G_{x_0} \geq 0$, which is
the positive cone of \cite{BlancoCasini2013}.  M\"obius transformations
have vanishing Schwarzian, so no anomaly arises.
\end{proof}

\begin{remark}[Higher dimensions]
\label{rem:highd}
The same decomposition holds for a ball of radius $R$ in a conformal
field theory of any dimension.  The Casini-Huerta-Myers weight
\cite{CHM2011} satisfies
\[
  2\pi\,\frac{R^2-r^2}{2R} \;=\; \pi R \;-\; \frac{\pi r^2}{R} ,
\]
so that
\begin{equation}
  \hat K_{\text{ball}} \;=\; \pi R\,H \;-\; \frac{\pi}{R}\,G ,
  \qquad G = \int r^2\,T_{00} ,
  \label{eq:highd}
\end{equation}
with $G$ the time component of the generator of special conformal
transformations, the higher-dimensional counterpart of the $G_{x_c}$ of
Proposition~\ref{prop:mobius}.  Blanco and Casini state their positive cone in general
dimension, obtaining it by transforming the Hamiltonian with an inversion
composed with a reflection so as to remain in the connected conformal
group \cite{BlancoCasini2013}; hence $G \geq 0$ and
$\hat K_{\text{ball}} \leq \pi R H$ in every dimension.  We claim no
novelty for this, any more than for
Proposition~\ref{prop:mobius}: what is special to two dimensions is not
the leading coefficient but the additive constant, which requires the
sharp quantum energy inequality of Section~\ref{sec:qei}.
\end{remark}

Two consequences organize the rest of the paper.  First, the coefficient
$\pi R/2$ is representation theory rather than a gift of any energy
inequality, and $c_0 \geq \pi R/2$ is exactly the threshold at which the
complementary weight $c_0 - 2\pi w$ enters the positive cone.  This
answers in advance the question of why $\pi R/2$ and not some other
number.  Second, any additive constant appearing for the one-sided $K_I$
is precisely the price of truncating the quadratic weight at the
endpoints of $I$.

\section{Quantum energy inequalities and admissible weights}
\label{sec:qei}

The tool for the truncation is the quantum energy inequality of Fewster
and Hollands \cite{FewsterHollands2005}, their Theorem 4.1: for a
conformal field theory with a single component $T$ of stress-energy and
any non-negative $G \in \mathscr{S}(\mathbb{R})$,
\begin{equation}
  \int G(v)\,\langle T(v)\rangle_\psi \,dv
  \;\geq\; -\,\frac{c}{12\pi}
  \int \Big(\tfrac{d}{dv}\sqrt{G(v)}\Big)^2 dv ,
  \qquad \psi \in \alg{D} ,
  \label{eq:QEI}
\end{equation}
with the convention $\frac{d}{dv}\sqrt G = 0$ where $G$ vanishes.  Three
features of their statement are used below.

\begin{remark}[Prefactor]
\label{rem:prefactor}
The constant $c/12\pi$ is for a \emph{single} chiral component.  For
$T_{00}$, which carries both chiralities, it doubles; the same inequality
appears in \cite{BCLR2018} with $c/6\pi$, using
$T_{00} = T_{++} + T_{--}$.  The two conventions are easily conflated and
the difference would be invisible in the final constant.
\end{remark}

\begin{remark}[Vanishing weights]
\label{rem:vanishing}
The weight used below vanishes at one interior point.  No special
treatment is needed: Remark (a) following Theorem 4.1 of
\cite{FewsterHollands2005}, with their Corollary A.2, establishes that
$\sqrt G$ lies in the Sobolev space $W^1(\mathbb{R})$ for non-negative
Schwartz $G$, that the
convention above is the weak derivative, and that the right-hand side of
\eqref{eq:QEI} is finite even for weights that are not strictly positive.
\end{remark}

\begin{remark}[Operator form]
\label{rem:operator}
Their Remark (b) states that, $\alg{D}$ being a
core for any smeared energy density, \eqref{eq:QEI} holds as an operator
inequality by standard quadratic form arguments.  The results below may
therefore be read as operator inequalities.
\end{remark}

The weights we use are not Schwartz, being constant at infinity, and are
not smooth, having corners where the interval meets the constant.  The
following lemma closes both gaps.

\begin{lemma}[Admissible weights]
\label{lem:admissible}
Let $g : \mathbb{R} \to [0,\infty)$ be Lipschitz, equal to a constant
$c_0 > 0$ outside a bounded set, smooth apart from finitely many
corners, and bounded below by some $\delta > 0$.  Then \eqref{eq:QEI} holds for
$g$.  The conclusion persists for $\delta = 0$ when $g$ is the limit of a
family $g + \delta$ of such weights differing by an additive constant.
\end{lemma}

\begin{proof}
\emph{Smoothness.}  Let $g_\epsilon = g * \eta_\epsilon$ with
$\eta_\epsilon$ a standard mollifier.  Then $g_\epsilon$ is smooth,
equals $c_0$ outside an $\epsilon$-neighborhood of that set,
satisfies
$g_\epsilon \geq \delta$, and $\|g_\epsilon - g\|_\infty \to 0$ since $g$
is Lipschitz.  As both are bounded below by $\delta$, the square root is
Lipschitz on their common range and $(\sqrt{g_\epsilon})'$ converges to
$(\sqrt g)'$ in $L^2$.

\emph{Decay.}  Let $\theta_L$ be smooth, equal to $1$ on
$[-L,L] \supset I$, supported in $[-2L,2L]$, with $\sqrt{\theta_L}$
smooth and $|(\sqrt{\theta_L})'| = O(1/L)$.  Then
$g_{\epsilon,L} = g_\epsilon\theta_L$ is smooth, non-negative and
compactly supported, hence Schwartz, and \eqref{eq:QEI} applies to it.

\emph{Limits.}  On the left,
\[
  \Big| \int (g_{\epsilon,L} - g)\,\langle T\rangle_\psi \Big|
  \;\leq\; \|g_\epsilon - g\|_\infty\,\Lambda_\psi
  \;+\; c_0\!\int_{|v|>L} \big|\langle T(v)\rangle_\psi\big|\,dv ,
\]
the first term vanishing as $\epsilon \to 0$ and the second as
$L \to \infty$ by the $O(v^{-4})$ decay.  Integrability of
$\langle T\rangle_\psi$ is essential here and not merely convenient: the
cutoff error equals $-c_0$ far from $I$ and is never small in supremum
norm, so it is controlled by the decay of the state rather than by the
closeness of the weights.  On the right, expanding
$\sqrt{g_{\epsilon,L}} = \sqrt{g_\epsilon}\,\sqrt{\theta_L}$ gives three
terms,
\[
  \Big(\big(\sqrt{g_\epsilon}\big)'\Big)^2\theta_L
  \;+\; 2\big(\sqrt{g_\epsilon}\big)'\sqrt{\theta_L}\,
    \sqrt{g_\epsilon}\,\big(\sqrt{\theta_L}\big)'
  \;+\; g_\epsilon\Big(\big(\sqrt{\theta_L}\big)'\Big)^2 ,
\]
and each simplifies exactly.  The middle term vanishes pointwise, since
$(\sqrt{\theta_L})'$ is supported where $g_\epsilon \equiv c_0$ and hence
$(\sqrt{g_\epsilon})' = 0$ there.  In the first, $\theta_L = 1$ wherever
$(\sqrt{g_\epsilon})' \neq 0$; in the third, $g_\epsilon = c_0$ wherever
$(\sqrt{\theta_L})' \neq 0$.  Hence
\[
  \int \Big(\tfrac{d}{dv}\sqrt{g_{\epsilon,L}}\Big)^2
  \;=\; \int \Big(\tfrac{d}{dv}\sqrt{g_\epsilon}\Big)^2
  \;+\; c_0\!\int \big((\sqrt{\theta_L})'\big)^2 ,
\]
with no error term, the second integral being $O(1/L)$.  Taking $\epsilon \to 0$ and
then $L \to \infty$ gives the claim.

For $\delta = 0$, let $g_\delta = g + \delta$.  The left side converges by
dominated convergence against $\Lambda_\psi$.  On the right,
$g_\delta' = g'$, so $\int (g')^2/4g_\delta$ increases monotonically to
$\int (g')^2/4g$ as $\delta \downarrow 0$, and by
Remark~\ref{rem:vanishing} the limit is finite and equals
$\int((\sqrt g)')^2$.
\end{proof}

\begin{remark}
The additive form of the last clause is not a convenience.  For a general
pointwise-decreasing family the derivatives need not converge to $g'$ at
all, and the monotonicity fails; for an additive shift only the
denominator moves.  The family used in Section~\ref{sec:top} is an
additive shift, globally and not merely on $I$.
\end{remark}

\section{The energy of the optimizer}
\label{sec:optimizer}

Both bounds below are proved by a one-parameter family of weights, and in
each case the family is an additive shift.  The following identity, which
holds for both and by the same two-line computation, is what turns the
bounds into equalities.

Throughout, $\psi_\lambda$ denotes the Fewster-Hollands optimizer for the
weight with parameter $\lambda$, characterized by
$\langle T\rangle_{\psi_\lambda} = \frac{c}{12\pi}(\sqrt g)''/\sqrt g$
and approached in the limit described in \cite{FewsterHollands2005}.

\begin{lemma}[Envelope identity]
\label{lem:envelope}
Let $\{g_\lambda\}$ be a family of weights admissible in the sense of
Lemma~\ref{lem:admissible} and differing by an additive constant,
$\partial_\lambda g_\lambda = 1$, each equal to a constant off one and the
same bounded interval.  Write
\begin{equation}
  B(\lambda) \;=\; \frac{c}{12\pi}\int_{\mathbb{R}}
  \Big(\big(\sqrt{g_\lambda}\big)'\Big)^2
  \label{eq:Blambda}
\end{equation}
for the corresponding Fewster-Hollands bound.  Then the optimizer
satisfies
\begin{equation}
  \big\langle P\big\rangle_{\psi_\lambda} \;=\; -\,B'(\lambda)
  \label{eq:envid}
\end{equation}
identically in $\lambda$.
\end{lemma}

\begin{proof}
Write $g = g_\lambda$.  Since $\partial_\lambda g = 1$ while $g'$ does not
depend on $\lambda$, and $\big((\sqrt g)'\big)^2 = (g')^2/4g$, we may
differentiate \eqref{eq:Blambda} under the integral to get
\begin{equation}
  B'(\lambda) \;=\; -\,\frac{c}{12\pi}\int \frac{(g')^2}{4g^2} .
  \label{eq:Bprime}
\end{equation}
On the other side, the optimizer has
$\langle T\rangle_{\psi_\lambda} = \frac{c}{12\pi}(\sqrt g)''/\sqrt g$, and
\[
  \frac{(\sqrt g)''}{\sqrt g}
  \;=\; \frac{g''}{2g} \;-\; \frac{(g')^2}{4g^2} ,
\]
so that, adding \eqref{eq:Bprime},
\begin{equation}
  \big\langle P\big\rangle_{\psi_\lambda} + B'(\lambda)
  \;=\; \frac{c}{12\pi}\int\left[\frac{g''}{2g}-\frac{(g')^2}{2g^2}\right]
  \;=\; \frac{c}{24\pi}\int_{\mathbb{R}}\left(\frac{g'}{g}\right)' .
  \label{eq:telescope}
\end{equation}
The last integrand is a total derivative.  Off a bounded set $g$ is a
positive constant, so $g'/g$ vanishes at both ends of the line and the
integral is zero.  This proves \eqref{eq:envid}.
\end{proof}

\begin{remark}
\label{rem:atoms}
The weights of Sections \ref{sec:top} and \ref{sec:bottom} have $g'$
discontinuous at the endpoints of $I$, so $g''$ carries an atom there.
Those atoms are not a correction to \eqref{eq:telescope} but part of it:
they are exactly what makes $g'/g$ telescope to zero.  Their weights are
computed where they are used, in \eqref{eq:deltas} and
Remark~\ref{rem:mirror}, and it is the vanishing of the modular weight at
the endpoints that makes them invisible to $\langle K_I\rangle$ while
contributing to $\langle P\rangle$.
\end{remark}

\section{The top of the range}
\label{sec:top}

\begin{theorem}
\label{thm:main}
Let $\alg{A}$ be a chiral conformal net of central charge $c$ and $I$ an
interval of length $R$.  Then for every $\beta > 0$,
\begin{equation}
  K_I \;\leq\; (1+\beta)\frac{\pi R}{2}\,P
  \;+\; \frac{c}{6}\Big[1 - \sqrt\beta\,\arctan\tfrac{1}{\sqrt\beta}\Big] ,
  \label{eq:family}
\end{equation}
and in the limit $\beta \to 0^+$,
\begin{equation}
  K_I \;\leq\; \frac{\pi R}{2}\,P \;+\; \frac{c}{6} .
  \label{eq:endpoint}
\end{equation}
Writing $M_{\max}(E,R)$ for the supremum of $\langle K_I\rangle_\psi$
over states with $\langle P\rangle_\psi \leq E$,
\begin{equation}
  M_{\max}(E,R)
  \;=\; \min_{\beta>0}\Big\{ (1+\beta)\tfrac{\pi R}{2}E
   + \tfrac{c}{6}\big[1-\sqrt\beta\arctan\tfrac1{\sqrt\beta}\big]\Big\}
  \;=\; \frac{\pi R}{2}E + \frac{c}{6}
   - \frac{\pi c^{2}}{288\,RE} + O\big((RE)^{-2}\big) ,
  \label{eq:Mmax}
\end{equation}
and neither constant can be improved:
\begin{equation}
  \sup_\psi \Big[ \langle K_I\rangle_\psi
   - \tfrac{\pi R}{2}\langle P\rangle_\psi \Big] \;=\; \frac{c}{6} .
  \label{eq:sharp}
\end{equation}
\end{theorem}

\begin{proof}[Proof of the inequalities]
Set $c_0 = (1+\beta)\pi R/2$ and
\begin{equation}
  g(x) \;=\; c_0 \;-\; 2\pi w(x)\,\chi_I(x) ,
  \label{eq:gdef}
\end{equation}
equal to $c_0$ off $I$, continuous since $w$ vanishes at the endpoints,
non-negative by \eqref{eq:wmax}, and satisfying
$\int g\,T = c_0 P - K_I$.  In the scaled coordinate
$u = (x-x_c)/R \in [-\frac12,\frac12]$ one has $w = R(\frac14 - u^2)$ and
\begin{equation}
  g \;=\; \frac{\beta\pi R}{2} + 2\pi R\,u^2 ,
  \label{eq:gscaled}
\end{equation}
so that, using $\frac{d}{dx} = \frac1R\frac{d}{du}$ and $dx = R\,du$,
\begin{equation}
  \int \Big(\tfrac{d}{dx}\sqrt g\Big)^2 dx
  \;=\; \int_{-1/2}^{1/2} \frac{(\partial_u g)^2}{4g}\,\frac{du}{R}
  \;=\; 2\pi\Big[1 - \sqrt\beta\,\arctan\tfrac{1}{\sqrt\beta}\Big] ,
  \label{eq:integral}
\end{equation}
independent of $R$, the factors of $R$ canceling.  The weight
\eqref{eq:gdef} is admissible in the sense of
Lemma~\ref{lem:admissible}, being bounded below by $\beta\pi R/2 > 0$, so
\eqref{eq:QEI} gives \eqref{eq:family}.

At $\beta = 0$ the weight is $g_0 = 2\pi R u^2$, which vanishes at the
midpoint of $I$.  By Remark~\ref{rem:vanishing} this causes no
difficulty.  The weight is constant at infinity and so not Schwartz, and
Corollary A.2 of \cite{FewsterHollands2005} does not apply verbatim; but
$\sqrt{g_0} = \sqrt{2\pi R}\,|u|$ is Lipschitz, and the argument of that
corollary identifies its weak derivative with the pointwise one off the
zero, a set of measure zero.  That derivative has constant modulus
$|(\sqrt{g_0})'| = \sqrt{2\pi/R}$, so the integral is
$(2\pi/R)\cdot R = 2\pi$ and the constant is
$\frac{c}{12\pi}\cdot 2\pi = \frac{c}{6}$.  Since
$g_\beta = g_0 + \beta\pi R/2$ globally, the second clause of
Lemma~\ref{lem:admissible} applies and gives \eqref{eq:endpoint}.
Finally $K_I \leq \lambda P + \mu$ gives
$\langle K_I\rangle \leq \lambda E + \mu$ on the stated states.
\end{proof}

\begin{proof}[Proof of exactness and sharpness]
Fewster and Hollands prove \eqref{eq:QEI} sharp: the right-hand side is
the infimum of the left over $\psi \in \alg{D}$, approached by
the vectors $U(\rho)^{-1}\Omega$, where $\rho$ is the circle
diffeomorphism whose induced reparametrization $V$ of the line satisfies
$V'(v) = 1/g(v)$, for which
$\langle T\rangle_\psi = \frac{c}{12\pi}(\sqrt g)''/\sqrt g$.  For our $g$
this optimizer is usable rather than formal, for two reasons.

First, off $I$ the weight is constant, so $(\sqrt g)'' = 0$ there and the
optimizer's energy density is supported in $\bar I$; the cutoff of
Lemma~\ref{lem:admissible} therefore costs it nothing.  Second, $g'$ has
a jump at each endpoint of $I$.  It vanishes outside $I$ and, by
\eqref{eq:gdef}, equals $-2\pi$ just inside $a$ and $+2\pi$ just inside
$b$; crossing either endpoint from left to right therefore changes $g'$
by $-2\pi$, and $(\sqrt g)'$ by $-\pi/\sqrt{c_0}$, at \emph{both}.
Hence $(\sqrt g)''$ carries a negative delta of that weight at each, so
that
\begin{equation}
  \langle T\rangle_\psi \Big|_{\text{endpoint}}
  \;=\; -\,\frac{c}{12\,c_0}\,\delta
  \qquad\text{at each of } a, b .
  \label{eq:deltas}
\end{equation}
Since $w(a) = w(b) = 0$, these contribute $-c/(6c_0)$ to
$\langle P\rangle$ and nothing at all to $\langle K_I\rangle$.  The
negative energy that pays for the constant sits exactly where the modular
weight vanishes.

Each $\beta$ therefore gives an upper bound attained in the limit by a
state of energy $E(\beta) = \langle P\rangle_{\psi_\beta}$, and by
Lemma~\ref{lem:envelope} that energy is exactly
$-\frac{2}{\pi R}\frac{d}{d\beta}\big[\frac{c}{6}b(\beta)\big]$, which is
the first-order condition for the minimization in \eqref{eq:Mmax}.
Hence the state saturating at parameter $\beta$ has precisely the energy
at which $\beta$ is optimal, the minimum over $\beta$ is attained, and
\eqref{eq:Mmax} is an equality rather than a bound.

For the rate, expand
\begin{equation}
  b(\beta) \;=\; 1 - \frac{\pi}{2}\sqrt\beta + \beta
   - \frac{\beta^2}{3} + O(\beta^3) ,
  \label{eq:bexpansion}
\end{equation}
obtained from $\arctan(1/t) = \frac{\pi}{2} - t + \frac{t^3}{3} - \cdots$
with $t = \sqrt\beta$.  Writing $\mu = \pi RE/2$ and minimizing
$\mu(1+t^2) + \frac{c}{6}(1 - \frac{\pi}{2}t + t^2)$ over $t$ gives
\begin{equation}
  t^* \;=\; \frac{\pi c}{4(3\pi RE + c)} ,
  \qquad
  \beta^* = (t^*)^2 \;\simeq\; \Big(\frac{c}{12\,RE}\Big)^2 ,
  \label{eq:betastar}
\end{equation}
at which the value is
$\mu + \frac{c}{6} - \dfrac{\pi^2c^2}{96\,(3\pi RE+c)}$.
Since $\pi^2c^2/\!\left[96(3\pi RE + c)\right]
= \pi c^2/(288\,RE) + O((RE)^{-2})$, this is \eqref{eq:Mmax}.

Finally the two constants are optimal.  From \eqref{eq:bexpansion},
$b'(\beta) = -\frac{\pi}{4\sqrt\beta} + 1 + O(\beta)$, so the envelope
relation gives $E(\beta) = \frac{c}{12R\sqrt\beta} - \frac{c}{3\pi R}
+ O(\sqrt\beta)$, which diverges as $\beta \to 0^+$.  Hence
\[
  \langle K_I\rangle_\psi - \tfrac{\pi R}{2}\langle P\rangle_\psi
  \;=\; \beta\,\tfrac{\pi R}{2}E(\beta) + \tfrac{c}{6}b(\beta)
  \;=\; \frac{c}{6} - \frac{\pi c}{24}\sqrt\beta + O(\beta) ,
\]
which increases to $c/6$ from below as $\beta \to 0^+$, giving
\eqref{eq:sharp}; and since $\sqrt{\beta^*} \simeq c/12RE$ the deficit is
$\pi c^2/(288\,RE)$, consistent with \eqref{eq:Mmax}.  The coefficient
$\pi R/2$ is optimal because $E(\beta) \to \infty$ along the same family
while $\langle K_I\rangle/(\pi RE/2) \to 1$, so any smaller coefficient
makes the supremum infinite.
\end{proof}

\begin{remark}[Legendre structure]
\label{rem:legendre}
Both \eqref{eq:Mmax} and its counterpart in Section~\ref{sec:bottom} have
the form $\min_\lambda[\lambda E + B(\lambda)]$ with $\lambda$ the
additive shift of the weight.  This is a Legendre transform, and it
explains an otherwise striking coincidence: the $\beta^*$ minimizing
\eqref{eq:Mmax} and the $\beta$ labeling the state that saturates at
energy $E$ are equal.  That is the envelope theorem, automatic once the
family is an additive shift and the bound is tight at each $\beta$, and
not a numerical accident; it is used in the proof of
Theorem~\ref{thm:main} in the form
$\frac{\pi R}{2}E = -\frac{c}{6}b'(\beta)$.  The one-parameter family is therefore the
structure of the answer rather than a device for obtaining it.
\end{remark}

\begin{remark}[The endpoint is an asymptote]
\label{rem:asymptote}
The bracket in \eqref{eq:family} behaves as
$1 - \frac{\pi}{2}\sqrt\beta + \beta + \cdots$ for small $\beta$, so the
derivative of the right side of \eqref{eq:Mmax} tends to $-\infty$ as
$\beta \to 0^+$ and $\beta = 0$ is never a local minimum.  The interior
optimum lies at $\beta^* \simeq (c/12RE)^2$ and beats the endpoint by
$\pi c^2/(288\,RE)$ at every finite $RE$.  Equation \eqref{eq:endpoint}
is the $RE \to \infty$ asymptote of the family, not a member of it that
wins at large $RE$.
\end{remark}

\begin{remark}[Where the negative energy sits]
\label{rem:where}
That the saturating state places its negative energy at the two endpoints
of $I$, and that the additive constant is the price of truncating the
M\"obius weight there, are the same fact.
Proposition~\ref{prop:mobius} says the untruncated weight needs no
constant at all.
\end{remark}

\begin{remark}[Small $RE$]
\label{rem:smallRE}
For large $\beta$ the bracket behaves as $1/(3\beta)$, so \eqref{eq:Mmax}
becomes $\min_\beta[(1+\beta)\frac{\pi R}{2}E + \frac{c}{18\beta}]$,
optimized at $\beta^* = \frac13\sqrt{c/\pi RE}$ with the two
$\beta$-dependent terms contributing $\frac16\sqrt{\pi cRE}$ each.  Hence
\begin{equation}
  M_{\max}(E,R) \;\longrightarrow\;
  \frac{\pi R}{2}E + \frac13\sqrt{\pi c R E}
  \qquad (RE \to 0) .
  \label{eq:smallRE}
\end{equation}
\end{remark}

\begin{corollary}[A sharp Bekenstein bound]
\label{cor:bekenstein}
Let $\psi$ be a state of mean energy at most $E$ and let $\Delta S_I$
denote the vacuum-subtracted entropy of the interval.  Then
\begin{equation}
  \Delta S_I \;\leq\; \frac{\pi R}{2}\,E \;+\; \frac{c}{6} .
  \label{eq:bekenstein}
\end{equation}
\end{corollary}

\begin{proof}
Positivity of the relative entropy of $\psi$ with respect to the vacuum
on $\alg{A}(I)$ is the statement $\Delta S_I \leq \Delta\langle K_I
\rangle$, which is Casini's form of the Bekenstein bound
\cite{Casini2008}.  Since $\langle K_I\rangle_\Omega = 0$ by
Remark~\ref{rem:normalization}, $\Delta\langle K_I\rangle =
\langle K_I\rangle_\psi$, and \eqref{eq:Mmax} bounds it by the right-hand
side of \eqref{eq:bekenstein}.
\end{proof}

\begin{remark}
\label{rem:bekremark}
Casini observes of his modular form $\Delta S \leq \Delta\langle K
\rangle$ that it ``does not introduce new physical constraints
semiclassically'', since it holds automatically and nothing further
bounds the modular energy.
Theorem~\ref{thm:main} supplies exactly that missing bound in the chiral
case, and \eqref{eq:bekenstein} is Bekenstein's original form
$S \lesssim ER$ with an explicit coefficient and an explicit additive
constant.  Two qualifications.  The coefficient $\pi R/2$ and the
constant $c/6$ are optimal for the intermediate quantity
$\langle K_I\rangle$, by \eqref{eq:sharp}; whether
\eqref{eq:bekenstein} itself is saturated depends on how small the
relative entropy can be made at fixed modular energy, which we do not
determine.  And $E$ here is the expectation of the chiral translation
generator $P$, not of a two-component Hamiltonian, so comparison with the
usual statement requires matching conventions.
\end{remark}

\section{The bottom of the range}
\label{sec:bottom}

There is no quantum energy inequality for $K_I$ alone: the weight
$2\pi w\chi_I$ has $((\sqrt g)')^2 \sim 1/(x-a)$ at the endpoints and the
integral diverges logarithmically.  The repair is the weight
$A + 2\pi w\chi_I$ with $A>0$, non-negative, constant off $I$, and
bounded below by $A$, hence admissible by Lemma~\ref{lem:admissible} and
in fact better behaved than \eqref{eq:gdef}, which vanishes at the
midpoint.

\begin{theorem}
\label{thm:bottom}
With $s = A/(\pi R/2)$,
\begin{equation}
  A\,P + K_I \;\geq\; -\frac{c}{12\pi}\,I_{\rm low}(s) ,
  \qquad
  I_{\rm low}(s) = 2\pi\Big[\sqrt{1+s}\,\artanh\tfrac{1}{\sqrt{1+s}} - 1\Big] ,
  \label{eq:Ilow}
\end{equation}
independent of $R$.  Consequently the infimum of $\langle K_I\rangle$
over states of mean energy at most $E$ is
\begin{equation}
  M_{\min}(E,R)
  \;=\; -\min_{A>0}\Big[AE + \frac{c}{12\pi}I_{\rm low}\big(A/(\pi R/2)\big)\Big]
  \;\simeq\; -\frac{c}{12}\Big[\ln\frac{24\pi RE}{c} - 1\Big]
  \label{eq:Mmin}
\end{equation}
as $RE \to \infty$, while $|M_{\min}| \to \frac13\sqrt{\pi cRE}$ as
$RE \to 0$.
\end{theorem}

\begin{proof}
Put $c_0 = A + \pi R/2$ and $k = \sqrt{1+s}/2 \geq \tfrac12$.  In the
scaled coordinate $u$ the weight reads
\begin{equation}
  g \;=\; A + 2\pi R\Big(\tfrac14 - u^2\Big)
  \;=\; 2\pi R\,\big(k^2 - u^2\big) ,
  \label{eq:glow}
\end{equation}
non-negative on $[-\frac12,\frac12]$ since $k \geq \frac12$, equal to $A$
at $u = \pm\frac12$ and hence continuous with the constant outside, and
bounded below by $A > 0$, so admissible by
Lemma~\ref{lem:admissible}; it is in fact better behaved than
\eqref{eq:gdef}, never vanishing.  With $\partial_u g = -4\pi Ru$,
\[
  \int \Big(\tfrac{d}{dx}\sqrt g\Big)^2 dx
  \;=\; \int_{-1/2}^{1/2}\frac{(\partial_u g)^2}{4g}\,\frac{du}{R}
  \;=\; 2\pi\!\int_{-1/2}^{1/2}\frac{u^2}{k^2-u^2}\,du ,
\]
again independent of $R$.  Writing
$\frac{u^2}{k^2-u^2} = -1 + \frac{k^2}{k^2-u^2}$ and using
$\int\frac{du}{k^2-u^2} = \frac1k\artanh\frac uk$ gives
$2\pi[2k\artanh\frac{1}{2k} - 1]$, which is $I_{\rm low}(s)$ on
substituting $2k = \sqrt{1+s}$.  Since $\int g\,T = A\,P + K_I$, the
inequality \eqref{eq:QEI} is \eqref{eq:Ilow}, and evaluating on states
with $\langle P\rangle \leq E$ and optimizing over $A$ gives
$M_{\min} \geq -\min_A[\,\cdot\,]$.

Attainment is the mirror of the argument in Section~\ref{sec:top}.  The
weight is constant off $I$, so the optimizer's energy density vanishes
there and the cutoff of Lemma~\ref{lem:admissible} costs it nothing; and
by Remark~\ref{rem:mirror} its endpoint deltas contribute to
$\langle P\rangle$ but not to $\langle K_I\rangle$, $w$ vanishing at $a$
and $b$.  So each $A$ gives a lower bound attained in the limit by a
state of energy $\langle P\rangle_{\psi_A}$, which by
Lemma~\ref{lem:envelope} equals
$-\frac{d}{dA}\big[\frac{c}{12\pi}I_{\rm low}(2A/\pi R)\big]$, precisely
the first-order condition for the minimization.  The state saturating at
parameter $A$ therefore has the energy at which $A$ is optimal, and the
first equality in \eqref{eq:Mmin} follows.

For the asymptotics, put $z = 1/\sqrt{1+s}$.  Then
$I_{\rm low} = 2\pi\big[\tfrac1z\artanh z - 1\big] = 2\pi\big[\tfrac{z^2}{3}
+ \tfrac{z^4}{5} + \cdots\big]$, so $I_{\rm low}(s) \simeq 2\pi/3s$ for
large $s$; and as $s \to 0^+$, $\artanh z = \frac12\ln\frac{1+z}{1-z}$
with $1-z \simeq s/2$ gives
$I_{\rm low}(s) = \pi\ln\frac4s - 2\pi
+ \frac{\pi s}{2}\big(\ln\frac4s + 1\big) + O(s^2)$, the logarithmic
divergence that the naive weight $2\pi w\chi_I$ suffers outright.

Large $RE$ selects small $s$.  With $s = 2A/\pi R$ the objective is
$AE + \frac{c}{12}\big[\ln\frac{2\pi R}{A} - 2\big] + O(s\ln s)$,
minimized at
$A^* = c/12E$, where it takes the value
$\frac{c}{12}\big[\ln\frac{24\pi RE}{c} - 1\big]$.  Small $RE$ selects
large $s$, where the objective is $AE + \frac{\pi cR}{36A}$, minimized at
$A^* = \frac16\sqrt{\pi cR/E}$ with value $\frac13\sqrt{\pi cRE}$.
\end{proof}

\begin{remark}[The mirror saturating state]
\label{rem:mirror}
The optimizer for \eqref{eq:glow} is the mirror of the one in
Theorem~\ref{thm:main}.  Here $g'$ equals $+2\pi$ just inside $a$ and $-2\pi$ just inside $b$, so
that crossing either endpoint changes it by $+2\pi$, and $(\sqrt g)''$
carries a \emph{positive} delta of weight $\pi/\sqrt A$ at each and
$\langle T\rangle_\psi = +\frac{c}{12A}\delta$ there, while
$(\sqrt g)'' < 0$ throughout the interior, $g$ being concave.  The state
is a smooth negative bulk inside $I$ with positive spikes at the
entangling points, and as before the spikes contribute to
$\langle P\rangle$ and not to $\langle K_I\rangle$.
\end{remark}

\begin{remark}[The shape of the range]
\label{rem:range}
The range is linear above and logarithmic below.  The two ends meet in
the small-$RE$ limit, where both $M_{\max} - \frac{\pi R}{2}E$ and
$|M_{\min}|$ approach $\frac13\sqrt{\pi cRE}$: the brackets of
\eqref{eq:Mmax} and \eqref{eq:Ilow} share the tail $1/(3t)$.  Both
quantities depend on $R$ and $E$ only through $RE$, and on $c$ only by
overall scale, as conformal invariance requires.
\end{remark}

\begin{lemma}[Dominance]
\label{lem:dominance}
$M_{\max}(E,R) > |M_{\min}(E,R)|$ for every $RE > 0$.
\end{lemma}

\begin{proof}
Write $\mu = \pi RE/2$.  Both quantities are Legendre minima in $\mu$,
\[
  M_{\max} = \min_{\beta>0}\Big[\mu(1+\beta) + \tfrac{c}{6}b(\beta)\Big] ,
  \qquad
  |M_{\min}| = \min_{s>0}\Big[\mu s + \tfrac{c}{12\pi}I_{\rm low}(s)\Big] ,
\]
and both vanish at $\mu = 0$, since $b(\beta) \to 0$ and
$I_{\rm low}(s) \to 0$ as their arguments diverge.  By the envelope
theorem $\frac{d}{d\mu}M_{\max} = 1 + \beta^*$ and
$\frac{d}{d\mu}|M_{\min}| = s^*$, so it suffices to show
$s^*(\mu) < 1 + \beta^*(\mu)$ for every $\mu > 0$.

Substitute $\beta = 4k^2$ and $s = 4a^2 - 1$, so that $k > 0$ and
$a > \frac12$.  The two integrals of Sections \ref{sec:top} and
\ref{sec:bottom} become
\[
  b(\beta) = \int_{-1/2}^{1/2}\frac{u^2}{k^2+u^2}\,du ,
  \qquad
  \frac{I_{\rm low}(s)}{2\pi} = \int_{-1/2}^{1/2}\frac{u^2}{a^2-u^2}\,du ,
\]
and differentiating under the integral, with
$dk/d\beta = 1/8k$ and $da/ds = 1/8a$, the two first-order conditions
take the matched form
\begin{equation}
  \mu \;=\; \frac{c}{24}\,J_+(k)
  \;=\; \frac{c}{24}\,J_-(a) ,
  \qquad
  J_\pm(\lambda) = \int_{-1/2}^{1/2}
  \frac{u^2}{(\lambda^2 \pm u^2)^2}\,du .
  \label{eq:matched}
\end{equation}
We write also
\begin{equation}
  K_\pm(\lambda) \;=\; \int_{-1/2}^{1/2}
  \frac{u^2}{(\lambda^2 \pm u^2)^3}\,du ,
  \qquad\text{so that}\qquad
  J_\pm'(\lambda) = -\,4\lambda\,K_\pm(\lambda) ,
  \label{eq:Kdef}
\end{equation}
these being the quantities that control the derivative of the matching.
In these variables $s < 1 + \beta$ reads $4a^2 - 1 < 1 + 4k^2$, that is
$a^2 < k^2 + \tfrac12$.  Both $J_\pm$ are strictly decreasing, so by
\eqref{eq:matched} this is equivalent to
\begin{equation}
  J_+(k) \;>\; J_-\!\big(\sqrt{k^2+\tfrac12}\,\big) .
  \label{eq:pointwise}
\end{equation}
Now compare the integrands.  For $|u| < \frac12$ we have
$0 < k^2 + u^2 < k^2 + \tfrac12 - u^2$, since the second inequality is
$2u^2 < \tfrac12$; and this holds on precisely the interval of
integration, with equality only at its endpoints.  Squaring and
inverting, the integrand on the left of \eqref{eq:pointwise} strictly
exceeds that on the right for $0 < |u| < \frac12$, and
\eqref{eq:pointwise} follows.
\end{proof}

\begin{remark}[The margin]
\label{rem:fourfifths}
The quantity $D = 1 + \beta^* - s^*$ admits a compact description.
Substitute $2k = \cot\theta$ and $2a = \coth\tau$, so that
$\beta = \cot^2\theta$ and $s = \operatorname{csch}^2\tau$, with
$\theta \in (0,\tfrac\pi2)$ and $\tau > 0$.  The two integrals of
\eqref{eq:matched}, and the bounds from which they came, then evaluate to
\begin{equation}
\begin{aligned}
  J_+ &= 2\tan\theta\Big(\theta - \tfrac{\sin 2\theta}{2}\Big) , &\qquad
  K_+ &= 2\tan^3\theta\Big(\theta - \tfrac{\sin 4\theta}{4}\Big) , \\
  J_- &= 2\tanh\tau\Big(\tfrac{\sinh 2\tau}{2} - \tau\Big) , &\qquad
  K_- &= 2\tanh^3\tau\Big(\tfrac{\sinh 4\tau}{4} - \tau\Big) ,
\end{aligned}
\label{eq:trighyp}
\end{equation}
in which the lower pair is the analytic continuation of the upper under
$\theta \mapsto -i\tau$, exactly for $J$ and up to sign for $K$.  In these
variables the matching condition is $J_+ = J_-$ and
\begin{equation}
  D \;=\; \csc^2\theta \;-\; \operatorname{csch}^2\tau .
  \label{eq:Dtrig}
\end{equation}

Both limits follow at once.  As $RE \to \infty$ we have
$\theta \to \tfrac\pi2$ and $\tau \to \infty$, so $D \to 1$.  As
$RE \to 0$ both tend to zero; expanding \eqref{eq:trighyp} gives
$J_+ = \tfrac43\theta^4 + \tfrac{8}{45}\theta^6 + O(\theta^8)$ and
$J_- = \tfrac43\tau^4 - \tfrac{8}{45}\tau^6 + O(\tau^8)$, so the matching
forces
\begin{equation}
  \tau \;=\; \theta\Big(1 + \tfrac{\theta^2}{15} + O(\theta^4)\Big) ,
  \label{eq:tautheta}
\end{equation}
and with $\csc^2\theta = \theta^{-2} + \tfrac13 + O(\theta^2)$ and
$\operatorname{csch}^2\tau = \tau^{-2} - \tfrac13 + O(\tau^2)$,
equation \eqref{eq:Dtrig} gives
$D \to \tfrac{2}{15} + \tfrac23 = \tfrac45$.  Hence
\begin{equation}
  M_{\max} - |M_{\min}| \;\simeq\; \tfrac45\,\mu \;=\; \tfrac{2\pi RE}{5}
  \qquad (RE \to 0) .
  \label{eq:margin}
\end{equation}

That $D$ rises monotonically between these two limits is
Theorem~\ref{thm:mono} below.  It is worth noting why the two arguments
already given do not reach it.  The pointwise comparison that settles
\eqref{eq:pointwise} yields exactly $a^2 - k^2 < \tfrac12$ and nothing
sharper, since $2u^2 < \tfrac12$ holds precisely on the interval of
integration; and no series argument will do, because by
\eqref{eq:trighyp} and \eqref{eq:tautheta} the difference $K_+ - K_-$
vanishes to twelfth order in $\theta$, the ratio
$(K_+-K_-)/\theta^{12}$ being
$\tfrac{2048}{7875}\big(1+\tfrac25\theta^2\big) + O(\theta^4)$.
What settles it is a change of variable.
\end{remark}

\begin{theorem}[Monotonicity of the margin]
\label{thm:mono}
$D = 1 + \beta^* - s^*$ is strictly increasing in $RE$.  Consequently
$D \in (\tfrac45, 1)$ throughout, and
\begin{equation}
  M_{\max} - |M_{\min}| \;>\; \tfrac45\,\mu \;=\; \tfrac{2\pi RE}{5}
  \qquad\text{for every } RE > 0 ,
  \label{eq:sharpmargin}
\end{equation}
which strengthens Lemma~\ref{lem:dominance}.
\end{theorem}

\begin{proof}
We first make explicit the implicit differentiation behind
Lemma~\ref{lem:dominance}.  Writing $\beta = 4k^2$ and $s = 4a^2-1$ and
using \eqref{eq:Kdef}, the two first-order conditions give
$dJ_+/d\beta = -K_+/2$ and $dJ_-/ds = -K_-/2$, so that along the matched
curve $J_+ = J_- = J$,
\begin{equation}
  \frac{dD}{dJ}
  \;=\; \frac{d\beta}{dJ} - \frac{ds}{dJ}
  \;=\; 2\left(\frac{1}{K_-} - \frac{1}{K_+}\right) .
  \label{eq:dDdJ}
\end{equation}
Since $J$ increases as $RE$ grows, $D$ is strictly increasing in $RE$ if
and only if $K_+ > K_-$ at matched $J$, and it is this that we prove.

Substitute $v = u^2$ in
\eqref{eq:matched}.  Since $du = dv/(2\sqrt v)$ and both integrands are
even, the four quantities become moments of the single measure
$d\rho(v) = \sqrt v\,dv$ on $(0,\tfrac14)$:
\begin{equation}
  J_+ = \int_0^{1/4}\!r_+^2\,d\rho , \quad
  K_+ = \int_0^{1/4}\!r_+^3\,d\rho , \quad
  J_- = \int_0^{1/4}\!r_-^2\,d\rho , \quad
  K_- = \int_0^{1/4}\!r_-^3\,d\rho ,
  \label{eq:stieltjes}
\end{equation}
where
\begin{equation}
  r_+(v) = \frac{1}{k^2+v} , \qquad r_-(v) = \frac{1}{a^2-v} ,
  \label{eq:albe}
\end{equation}
both positive on $(0,\tfrac14)$, the second because $a > \tfrac12$.  This
is the Stieltjes representation of the analytic function whose two real
branches carry the upper and lower families: $r_+$ and $r_-$ are its
resolvents at $-k^2$ and at $a^2$.

Now set
\begin{equation}
  g \;=\; \frac{r_+^2+r_+r_-+r_-^2}{r_++r_-} ,
  \qquad\text{so that}\qquad
  r_+^3 - r_-^3 \;=\; g\,\big(r_+^2-r_-^2\big) ,
  \label{eq:gdef2}
\end{equation}
the second identity being the factorization of the difference of cubes
divided by that of squares.  The point of the substitution is that $g$
satisfies an exact differential identity.  From \eqref{eq:albe},
$r_+' = -r_+^2$ and $r_-' = +r_-^2$, so
\[
  g' \;=\; \frac{(r_+^2+2r_+r_-)r_+' + (r_-^2+2r_+r_-)r_-'}
  {(r_++r_-)^2}
  \;=\; \frac{-(r_+^2-r_-^2)\big[(r_+^2+r_-^2)+2r_+r_-\big]}
  {(r_++r_-)^2} ,
\]
that is
\begin{equation}
  g' \;=\; -\big(r_+^2 - r_-^2\big) ,
  \qquad\text{whence}\qquad
  r_+^3-r_-^3 \;=\; -\,g\,g' \;=\; -\tfrac12\big(g^2\big)' .
  \label{eq:gprime}
\end{equation}

Let $d\nu(v) = dv/\sqrt v$ on $(0,\tfrac14)$, a \emph{probability} measure,
since $\int_0^{1/4}v^{-1/2}dv = 2\sqrt{v}\,\big|_0^{1/4} = 1$.  Integrating
by parts, and noting that the boundary term at $v=0$ vanishes because
$\sqrt v \to 0$ while $g$ stays bounded,
\begin{align}
  J_+ - J_- &= -\int_0^{1/4}\! g'\sqrt v\,dv
  \;=\; \tfrac12\Big[\textstyle\int g\,d\nu - g(\tfrac14)\Big] ,
  \label{eq:Jbyparts}\\
  K_+ - K_- &= -\tfrac12\int_0^{1/4}\!\big(g^2\big)'\sqrt v\,dv
  \;=\; \tfrac14\Big[\textstyle\int g^2\,d\nu - g(\tfrac14)^2\Big] .
  \label{eq:Kbyparts}
\end{align}
So the hypothesis $J_+ = J_-$ says exactly that the $\nu$-mean of $g$
equals its value at the right endpoint,
\begin{equation}
  \int g\,d\nu \;=\; g(\tfrac14) ,
  \label{eq:meanvalue}
\end{equation}
and the conclusion $K_+ > K_-$ says that the $\nu$-mean of $g^2$ exceeds
the square of that value.  By Cauchy-Schwarz for the probability measure
$\nu$,
\[
  \int g^2\,d\nu \;\geq\; \Big(\int g\,d\nu\Big)^{\!2}
  \;=\; g(\tfrac14)^2 ,
\]
with equality if and only if $g$ is $\nu$-almost everywhere constant.  By
\eqref{eq:gprime} that would require $r_+^2 = r_-^2$ throughout,
whereas $r_+ = r_-$ only at the single point
$v_0 = (a^2-k^2)/2$.  The inequality is therefore strict, and
$K_+ > K_-$.

Finally, $D$ is strictly increasing with limits $\tfrac45$ and $1$ by
Remark~\ref{rem:fourfifths}, so $\tfrac45 < D < 1$, and
$\frac{d}{d\mu}(M_{\max}-|M_{\min}|) = D$ with both vanishing at
$\mu = 0$ gives \eqref{eq:sharpmargin}.
\end{proof}

\begin{remark}
The mechanism is worth isolating.  The substitution $v = u^2$ turns both
families into moments of one measure, after which the ratio
\eqref{eq:gdef2} of the cubic to the quadratic difference obeys
$g' = -(r_+^2-r_-^2)$ exactly.  That single identity converts the
constraint and the conclusion into the first and second $\nu$-moments of
the same function, and what separates them is then nothing but the
positivity of a variance.  The twelfth-order degeneracy noted in
Remark~\ref{rem:fourfifths} is invisible here because no expansion is
taken.
\end{remark}

\begin{remark}[Numerical range versus compression]
\label{rem:compression}
Theorems \ref{thm:main} and \ref{thm:bottom} concern the numerical range
over the ball of states of mean energy at most $E$.  Writing $N_{\max}$
and $N_{\min}$ for the extremes over unit vectors in the range of the
spectral projection $P_E$, that subspace lies inside the ball, so
$N_{\max} \leq M_{\max}$ and $|N_{\min}| \leq |M_{\min}|$, whence
$\|P_EK_IP_E\| \leq M_{\max}$.  Equality would require
$N_{\max} = M_{\max}$, that is spectral concentration of the saturating
states, which we do not have.  The bound uses
Lemma~\ref{lem:dominance}, without which the right-hand side would read
$\max(M_{\max},|M_{\min}|)$.
\end{remark}

\section{Several intervals}
\label{sec:multi}

The constant $c/6$ was identified in Remark~\ref{rem:where} as the price
of truncating the M\"obius weight at the two endpoints of $I$.  If that
reading is correct the constant should be additive over entangling
points, and it is.

\begin{proposition}[Additivity over entangling points]
\label{prop:multi}
Let $I_1,\dots,I_n$ be pairwise disjoint intervals of equal length $R$,
let $w_i$ be the modular weight \eqref{eq:KI} of $I_i$, and put
$K = \sum_{i=1}^n K_{I_i}$.  Then
\begin{equation}
  K \;\leq\; \frac{\pi R}{2}\,P \;+\; \frac{n\,c}{6} ,
  \label{eq:multi}
\end{equation}
and the constant is $c/12$ per entangling point.
\end{proposition}

\begin{proof}
Take $g = c_0 - 2\pi\sum_i w_i\chi_{I_i}$ with $c_0 = \pi R/2$.  Since the
intervals are disjoint, at most one summand is active at each point, so
by \eqref{eq:wmax} the weight is non-negative and vanishes exactly at the
$n$ midpoints; and since each $w_i$ vanishes at the endpoints of $I_i$,
$g$ is continuous, equalling $c_0$ off $\bigcup_i I_i$.  It is Lipschitz, smooth apart from the $2n$ corners at the interval
endpoints, and constant off $\bigcup_i I_i$, so it is admissible by
Lemma~\ref{lem:admissible} for every positive additive shift; the weight
itself is the limit of that family as the shift tends to zero, and the
second clause of the lemma applies exactly as at $\beta = 0$ in the proof
of Theorem~\ref{thm:main}.  Also $\int g\,T = c_0 P - K$.

The right-hand side of \eqref{eq:QEI} localizes.  On $I_i$ the calculation
is that of Theorem~\ref{thm:main} at $\beta = 0$, contributing $2\pi$;
between and outside the intervals $g$ is constant and contributes
nothing.  Hence $\int((\sqrt g)')^2 = 2\pi n$ and the constant is
$\frac{c}{12\pi}\cdot 2\pi n = nc/6$, which is \eqref{eq:multi}.
\end{proof}

\begin{remark}
Two qualifications.  The object bounded is the sum of the one-sided
weights of the individual intervals, which for $n \geq 2$ differs from the
modular Hamiltonian of the union: the latter acquires a nonlocal piece,
known explicitly for the free fermion.  And equal lengths are used only to
make the weight vanish at every midpoint at once; for unequal lengths the
same argument with $c_0 = \pi R_{\max}/2$ gives
$\frac{c}{6}\sum_i b(\beta_i)$ with $\beta_i = R_{\max}/R_i - 1$, which
reduces to $nc/6$ when the lengths agree.  Neither qualification affects
the reading of the constant: it is $c/12$ per entangling point, and the
count is what changes.
\end{remark}

\section{Discussion}
\label{sec:disc}

\subsection*{Relation to existing bounds}

Blanco, Casini, Leston and Rosso \cite{BCLR2018} derive modular energy
inequalities from monotonicity of relative entropy and describe their
result as improving on \cite{FewsterHollands2005}.  The improvement is an
additive free-entropy term which they show vanishes for pure optimizers,
operating only for mixed states; and comparing the two bounds directly
they note that the Fewster-Hollands one is the more restrictive, since it
is sharp.  The saturating state of Theorem~\ref{thm:main} is a vector
state, so \eqref{eq:sharp} is unaffected.

Casini and Martinek \cite{CasiniMartinek2024} obtain an energy inequality
in which the smeared energy density is bounded below by an operator
rather than a number.  They describe \eqref{eq:QEI} as the most general
and sharpest of the numerical bounds and observe that neither of the two
is generally stronger.  Their result is for the chiral fermion
specifically, whereas Theorem~\ref{thm:main} holds for any chiral net
with a stress tensor.

To our knowledge the operator inequality \eqref{eq:endpoint} has not
appeared in the literature.  The nearest neighbors differ in the quantity bounded or in
the direction of the inequality.  Casini's route to the Bekenstein bound
\cite{Casini2008} bounds entropy by energy, not modular energy by the
translation generator.  \cite{BCLR2018} bound modular energy from below
and by entropy.  \cite{CasiniMartinek2024} bound energy density from
below by an operator.  The half-sided modular literature descending from
Borchers \cite{Borchers1992,Borchers1996,ArakiZsido2005,Borchers2000}
gives the exact commutation relation between modular flow and
translations rather than a quantitative inequality; we searched it
without finding one.  A negative result from the open literature is of
course weak evidence.

\subsection*{Outlook}

The natural extension is to higher dimensions, and
Remark~\ref{rem:highd} locates what is required.  The leading coefficient
comes free in every dimension, the M\"obius decomposition and the
positive cone being available there.  The additive constant,
however, is not merely unavailable but obstructed, and it is worth saying
why, since the obstruction identifies exactly what is special about the
two-dimensional case.

The constant would come from a quantum energy inequality for the weight
$g = \pi R - 2\pi w_{\text{ball}}\chi_{\text{ball}}$, with
$w_{\text{ball}} = (R^2-r^2)/2R$ as in Remark~\ref{rem:highd}, which is
non-negative, bounded, and smeared over a spacelike slice.  Quantum
energy inequalities for averaging along a timelike worldline hold in
every dimension \cite{FewsterEveson1998}.  Purely spatial ones do not:
Ford, Helfer and Roman \cite{FHR2002} exhibit states of a massless scalar
field carrying an arbitrarily large amount of negative energy in a given
region of space at a fixed time.  Their sampling is compactly supported,
where our $g$ has a constant tail $\pi R$ charging positive energy
anywhere at full weight, so their theorem does not by itself defeat the
tailed bound.  What it does establish is that the machinery the
two-dimensional proof rests on, a spatial inequality depending on the
weight alone, has no counterpart above two dimensions.  We expect the
stronger statement as well, since in their states the total energy grows
only logarithmically in the cutoff and at one higher power of the
amplitude than the sampled negative energy, so that the tail cannot
rescue the bound at small amplitude; we do not prove it here.  Their
field is also minimally rather than conformally coupled, and so is not
conformal in four dimensions.

The two-dimensional case escapes because its averaging is null rather
than spacelike.  A chiral coordinate is a null coordinate, so
\eqref{eq:QEI} is a null-averaged inequality, and averaging along a null
direction is a very different matter from averaging over a slice.  What
makes an interval work is a coincidence of two properties: its modular
weight is bounded, by \eqref{eq:wmax}, \emph{and} it is supported on a
null line.  Among the regions whose modular weights are known explicitly,
higher dimensions offer one property or the other but not both.  The ball
has a bounded weight, by Remark~\ref{rem:highd}, but a spacelike
smearing.  A wedge, or a null cut of a null plane, has a null smearing
but a weight growing linearly without bound, so that no inequality
$K \leq \lambda H + \mu$ can hold for it at any $\lambda$.  Nor would a
region combining both properties help: Fewster and Roman \cite{FR2003}
show that weighted averages of the null-contracted stress tensor along a
null geodesic are unbounded below in four dimensions, in contrast to two,
so above two dimensions the null direction supports no inequality either.
Two dimensions is thus special twice over, and it is why the constant
$c/6$ exists here with no evident counterpart elsewhere.

The several-interval case is settled by
Proposition~\ref{prop:multi}, and the answer supports reading $c/6$ as a
boundary quantity: the constant is $c/12$ per entangling point, and it is
the count of those points, not the geometry, that changes.  It is worth
recording that $c/12$ per point is also the coefficient of the chiral
entanglement entropy $\frac{c}{6}\log(R/\epsilon)$ under the same
decomposition.  We do not know whether that agreement is structural.  A
test would be a region whose entangling points are not those of a
disjoint union of intervals, for which the entropy coefficient and the
truncation count could differ.

\section*{Declarations}

\noindent
\textbf{Funding.}  No funding was received for this work.

\medskip
\noindent
\textbf{Competing interests.}  The authors declare no competing
interests.

\medskip
\noindent
\textbf{Data availability.}  No datasets were generated or analysed, and no
numerical computation underlies any result reported here.

\end{document}